\documentclass[11pt]{article}
\usepackage{amssymb,latexsym,amsmath, amsthm}
\usepackage{amsfonts, graphics}
\usepackage{hyperref}
\def\R{\mathbb R}

\def\N{\mathbb N}

\def\supp{\mathrm{supp}\,}

\def\T{{\mathcal T}}

\def\be{\begin{equation}}
\def\ee{\end{equation}}
\def\bea{\begin{eqnarray}}
\def\eea{\end{eqnarray}}
\def\beas{\begin{eqnarray*}}
\def\eeas{\end{eqnarray*}}

\newtheorem{theorem}{Theorem}[section]
\newtheorem{corollary}{Corollary}

\newtheorem{lemma}[theorem]{Lemma}
\newtheorem{proposition}{Proposition}

\theoremstyle{definition}
\newtheorem{definition}[theorem]{Definition}
\newtheorem{remark}{Remark}

\begin{document}

\title{On the transport operators arising from linearizing the
  Vlasov-Poisson or Einstein-Vlasov system about isotropic steady states}
\author{Gerhard Rein, Christopher Straub\\
        Fakult\"at f\"ur Mathematik, Physik und Informatik\\
        Universit\"at Bayreuth\\
        D-95440 Bayreuth, Germany\\
        gerhard.rein@uni-bayreuth.de, christopher.straub@uni-bayreuth.de}

\maketitle

\begin{abstract}
	If the Vlasov-Poisson or Einstein-Vlasov system is linearized about
	an isotropic steady state, a linear operator arises the properties of which
	are relevant in the linear as well as nonlinear stability analysis of the
	given steady state.
	We prove that when defined on a suitable Hilbert space and
	equipped with the proper domain of definition
	this transport operator ${\mathcal T}$ is skew-adjoint, i.e.,
	${\mathcal T}^\ast = - {\mathcal T}$.
	In the Vlasov-Poisson case we also determine the kernel
	of this operator.
\end{abstract}

\maketitle
\section{Introduction}
The Vlasov-Poisson and Einstein-Vlasov systems describe a self-gravitating,
collisionless gas in the framework of Newtonian mechanics or General Relativity,
respectively. For an important class of steady
states of these systems the particle
density on phase space depends only on the local or particle
energy, i.e., $f_0(x,v)=\phi(E(x,v))$ with some given function $\phi$.
The particle energy $E$ is a $C^2$ function of position $x\in {\mathbb R}^3$ and
velocity (or momentum) $v\in {\mathbb R}^3$, determined by the steady state;
details follow below. We study the linear transport operator 
\begin{equation} \label{Tclassdef}
{\mathcal T} f := \{f,E\} = {\partial_v} E \cdot {\partial_x} f - {\partial_x} E \cdot {\partial_v} f .
\end{equation}
Here $\{\cdot,\cdot\}$ denotes the usual Poisson bracket,
${\partial_x}$ and ${\partial_v}$ are gradients with respect to the indicated
variable, and $\cdot$ denotes the Euclidean scalar product.
The operator ${\mathcal T}$ arises in the stability analysis of the steady state
by linearization of the corresponding system.
It is skew-symmetric with respect to the $L^2$ scalar
product on the proper domain in phase space. As is amply documented in the
literature, it is in general not obvious how to realize such an operator
as skew-adjoint, which is the desired property from
a functional analysis point of view, cf.\ \cite[13.4 Example]{Rud}.
Since ${\mathcal T}$ has come up in various places in the literature
\cite{GuLi,HaLiRe,LeMeRa1}
without the above question having been properly addressed, we
give a careful proof that  ${\mathcal T}$ is skew-adjoint when defined
on a suitable Hilbert space and equipped
with the proper domain. In addition, we determine the kernel
of this operator in the Vlasov-Poisson case;
the analogous result in the relativistic case is open.

We now discuss the two systems in more detail.
As motivation and background we note that in astrophysics,
a self-gravitating, collisionless gas is used to model galaxies
or globular clusters, cf.~\cite{BT}.
In the context of Newtonian physics such an ensemble is described
by the Vlasov-Poisson system:
\begin{equation}
\partial_{t}f+v\cdot {\partial_x} f - {\partial_x} U\cdot {\partial_v} f =0,
\label{vlasov}
\end{equation}
\begin{equation}
\Delta U = 4\pi \rho,\ \lim_{|x| \to \infty} U(t,x) = 0, \label{poisson}
\end{equation}
\begin{equation}
\rho(t,x) = \int f(t,x,v)\,dv. \label{rhodef} 
\end{equation}
Here $t\in {\mathbb R},\ x,v\in {\mathbb R}^3$ stand for time, position, and velocity,
$f=f(t,x,v)\geq 0$ is the density of the ensemble on phase space,
$\rho=\rho(t,x)$ is the corresponding spatial density,
$U=U(t,x)$ denotes the induced gravitational potential, and unless
explicitly stated otherwise, integrals extend over ${\mathbb R}^3$.
For background on this system we refer to \cite{Rein07}.

When describing the same physical situation in the general relativistic
context the role of the potential is taken over by the Lorentz metric
on spacetime. We restrict
ourselves to the case of spherical symmetry
and write the metric in Schwarzschild form
\[
ds^2=-e^{2\mu(t,r)}dt^2 + e^{2\lambda(t,r)}dr^2+
r^2(d\theta^2+\sin^2\theta\,d\phi^2).
\]
Here $t\in{\mathbb R}$ is the time coordinate, $r\in [0, \infty[$ is the
area radius, i.e., $4 \pi r^2$ is the area of the orbit of the
symmetry group $\mathrm{SO}(3)$ labeled by $r$, and the angles
$\theta\in[0, \pi]$ and $\phi\in[0, 2\pi]$ parametrize these
orbits.
The spacetime is required to be asymptotically flat
with a regular center, which corresponds to the boundary conditions
\begin{equation} \label{bc}
\lim_{r\to \infty} \lambda(t,r)= \lim_{r\to \infty} \mu(t,r) = 0 = \lambda(t,0).
\end{equation}
We write
$(x^a)=r(\sin \theta \cos\phi, \sin \theta \sin \phi, \cos\theta)$,
let $(p^\alpha)$ denote the canonical momenta
corresponding to the spacetime coordinates $(x^\alpha) = (t,x^1,x^2,x^3)$,
and define
\[
v^a := p^a + (e^\lambda - 1) \frac{x\cdot p}{r} \frac{x^a}{r}
\ \mbox{for}\ a=1,2,3;
\]
as above $\cdot$ stands for the Euclidean scalar product in ${\mathbb R}^3$.
As in the Newtonian case all particles
have the same rest mass, normalized to unity.
Then 
\[
p_0 = - e^\mu \sqrt{1+|v|^2},\ \mbox{where}\ |v|^2 = v\cdot v.
\]
The density function $f=f(t,x,v)\geq 0$ is to be
\emph{spherically symmetric}, i.e., for any rotation $A\in \mathrm{SO}(3)$,
$f(t,x,v)=f(t,Ax,Av)$.
In this set-up the Einstein-Vlasov system becomes
\begin{equation} \label{vlasovr}
{\partial_t} f + e^{\mu - \lambda}\frac{v}{{\sqrt{1+|v|^2}}}\cdot {\partial_x} f -
\left( {\partial_t} \lambda \frac{x\cdot v}{r} + e^{\mu - \lambda} \partial_r\mu
{\sqrt{1+|v|^2}} \right) \frac{x}{r} \cdot {\partial_v} f =0,
\end{equation}
\begin{align*}
e^{-2\lambda} (2 r \partial_r\lambda -1) +1 
&=
8\pi r^2 \rho ,\\
e^{-2\lambda} (2 r \partial_r\mu +1) -1 
&= 
8\pi r^2 p, 
\end{align*}
\begin{align*}
\rho(t,r) 
&=
\rho(t,x) = \int {\sqrt{1+|v|^2}} f(t,x,v)\,dv ,\\
p(t,r) 
&= 
p(t,x) = \int \left(\frac{x\cdot v}{r}\right)^2
f(t,x,v)\frac{dv}{{\sqrt{1+|v|^2}}}. 
\end{align*}
For a detailed derivation of these
equations we refer to \cite{Rein95}; note that we have written down
only a closed subsystem of the over-determined Einstein-Vlasov
system. For more background on the latter we refer to \cite{And05}.

Both systems possess a plethora of steady states.
An important class of spherically symmetric
such states is obtained via the ansatz
\begin{equation} \label{ansatz}
f_0 = \phi\circ E,
\end{equation}
where $\phi\colon {\mathbb R} \to [0,\infty[$ is a suitable ansatz function
and $E$ is the local or particle energy, i.e.,
\begin{equation} \label{Edef}
E(x,v) = \left\{
\begin{array}{ll}
\displaystyle \frac{1}{2} |v|^2 + U_0(x),&
\ \mbox{Newtonian case},\\
\displaystyle - p_0 = e^{\mu_0(x)} \sqrt{1+|v|^2},&
\ \mbox{relativistic case}.
\end{array}
\right.
\end{equation}
Here $U_0$ or $\mu_0$ is the time independent
potential or metric component of the steady state, and it is easy to check
that $f_0$ satisfies the corresponding Vlasov equation
\eqref{vlasov} or \eqref{vlasovr}.
In order to obtain a steady state one can substitute the
ansatz \eqref{ansatz} into the definition of $\rho$ (and $p$) and try
to solve the resulting non-linear problem for $U_0$ or $\mu_0$;
we refer to \cite{RaRe} for sufficient conditions on the ansatz function
$\phi$ which guarantee that suitable steady states result from this
procedure.

We now fix such a steady state and linearize the Vlasov-Poisson system
about it. If we substitute $f_0 + f$ for $f$ and drop the quadratic term in
$f$ in the Vlasov equation \eqref{vlasov}, the linearized equation becomes
\[
{\partial_t} f + v\cdot {\partial_x} f - {\partial_x} U_0 \cdot {\partial_v} f + \ldots = 0,
\]
where $\ldots$ stands for a term which depends non-locally on $f$,
i.e., on the field induced by $f$, and which is of no interest here.
The transport term induced by the steady state
is ${\mathcal T} f = \{f, E\}$.
In a stability analysis,
the operator typically acts on the weighted space $L^2_\chi (S_0)$,
where $S_0=\{(x,v)\in {\mathbb R}^3 \times {\mathbb R}^3 \mid f_0(x,v)>0\}$ so that
$\bar S_0$ is the support of the steady state, a compact set, and
$\chi= 1/|\phi'\circ E|$; throughout the paper functions are real-valued.
The weight $\chi$ arises naturally
if one computes the second variation of the so-called energy-Casimir functional
at a given steady state, and this second variation is used
to define the distance measure for the stability estimate, cf.\ \cite{GuRe2007};
the results of the present paper remain true if we take $\chi=1$.
The crucial point for our analysis is to give the proper weak formulation of
the operator ${\mathcal T}$ and define a domain for it such that it
is not only skew-symmetric, but skew-adjoint.
We will also determine the kernel of this operator.

If we proceed analogously for the Einstein-Vlasov system, the linearized
equation reads
\begin{equation} \label{linvlr}
{\partial_t} f + e^{\mu_0 - \lambda_0}\frac{v}{{\sqrt{1+|v|^2}}}\cdot {\partial_x} f -
e^{\mu_0 - \lambda_0}
{\sqrt{1+|v|^2}}  {\partial_x} \mu_0 \cdot {\partial_v} f + \ldots =0;
\end{equation}
the analogous comment as in the Newtonian case applies to the term
$\ldots$. Here the linear transport operator equals $e^{- \lambda_0} {\mathcal T}$.
But at the same time the natural weight in the $L^2$ space on which
the operator acts is $\chi=e^{\lambda_0}/|\phi'\circ E|$. Hence
the true relativistic transport operator $e^{- \lambda_0} {\mathcal T}$
on the properly weighted space is skew-adjoint iff
${\mathcal T}$ is skew-adjoint on the space with weight
$\chi=1/|\phi'\circ E|$. Thus we keep to the operator ${\mathcal T}$ as defined before,
since this allows us to treat the Newtonian and the relativistic case
largely in parallel.

In the context of both linear and nonlinear stability results
the properties of the operator ${\mathcal T}$
play a role at various places
in the mathematics literature, cf.~\cite{GuLi,HaLiRe,LeMeRa1},
but also in the physics literature, cf.~\cite{IT68}, where ${\mathcal T}$
appears both in the Newtonian and the relativistic case.
We also refer to the second author's master thesis \cite{Str2019}.

In the next section we formulate the basic assumptions and main
results of this paper. The skew-adjointness of ${\mathcal T}$ is proven
in Section~3, and in the Newtonian case its kernel is determined in Section~4.

\section{Basic assumptions and main results}

\begin{itemize}
	\item[(A1)]
	Let $\phi \colon {\mathbb R} \to [0,\infty[$ be 
	such that there exists a constant
	$E_0$ such that $\phi \in C^1(]-\infty,E_0[)$ with 
	$\phi ' (E) < 0$ for $E < E_0$, and $\phi(E)=0$ for $E\geq E_0$.
	In the Newtonian case, $E_0<0$, in the relativistic case, $0<E_0<1$.
	\item[(A2)]
	Let $f_0=\phi\circ E$ be a steady state of the Vlasov-Poisson or
	Einstein-Vlasov system, where $E$ is defined as in \eqref{Edef}
	and with corresponding, spherically symmetric
	potential $U_0\in C^2({\mathbb R}^3)\cap C^2([0,\infty[)$ or metric components
	$\lambda_0, \mu_0\in C^2({\mathbb R}^3)\cap C^2([0,\infty[)$. Let these satisfy the
	boundary conditions specified in \eqref{poisson} or \eqref{bc} respectively.
\end{itemize}
At this point we identify $U_0(x) = U_0(|x|)$ if $U_0$ is spherically symmetric.
As shown in \cite{RaRe}, many steady states satisfy these assumptions.
In addition,
$U_0$ and $\mu_0$ are strictly increasing as radial functions.
The set where $f_0$ is positive is given by
\[
S_0 = \{ (x,v) \in {\mathbb R}^3 \times {\mathbb R}^3 \mid E(x,v) < E_0 \}.
\]
Due to the cut-off energy $E_0$ in (A1) and the boundary condition
for $U_0$ respectively $\mu_0$ at infinity $S_0$ is
a bounded, spherically symmetric domain; note that the restriction on
$E_0$ is such that by \eqref{Edef}, $E(x,v) \geq E_0$ in both
the Newtonian or the relativistic case, if $|x|$ or $|v|$ are sufficiently
large.
On functions $f\in C^1(S_0)$ we define ${\mathcal T} f$ as in \eqref{Tclassdef}.
The characteristic flow
of the first order differential operator ${\mathcal T}$ is given by
the Hamiltonian system
\begin{equation} \label{charsys}
\dot x = {\partial_v} E(x,v),\ \dot v = - {\partial_x} E(x,v).
\end{equation}
For each $(x,v) \in {\mathbb R}^3 \times {\mathbb R}^3$ it has a global, unique solution
$t\mapsto (X,V)(t,x,v)$
such that $(X,V)(0,x,v)=(x,v)$. For every $t\in{\mathbb R}$ the map
$(X,V)(t,\cdot,\cdot)$ is a $C^1$ diffeomorphism
on ${\mathbb R}^3 \times {\mathbb R}^3$ which is measure-preserving,
i.e., $\det \frac{\partial(X,V)}{\partial(x,v)}(t,x,v) =1$,
since the divergence of the
right hand side of \eqref{charsys} vanishes. The characteristic flow
preserves the particle energy $E$, in particular
the set $S_0$ is invariant under the flow.
Due to the spherical symmetry of the potential or metric
of the steady state the quantity
\[
L(x,v):= |x\times v|^2,
\]
the modulus of angular momentum
squared of the particle with coordinates $(x,v) \in {\mathbb R}^3 \times {\mathbb R}^3$,
is also preserved by the characteristic flow.
For $f\in C^1(S_0)$ we can express the transport operator as
\begin{equation} \label{Tinfgen}
{\mathcal T} f(x,v) = \frac{d}{dt}  \left[ f(X(t,x,v), V(t,x,v)) \right]\big|_{t=0},
\quad (x,v) \in S_0.
\end{equation}

\begin{remark}
	The characteristic flow corresponding to the linearized Einstein-Vlasov
	system as written in \eqref{linvlr} is not measure-preserving, due to the
	additional factor $e^{-\lambda_0}$ in the right hand side of the
	corresponding characteristic system. However, this flow does preserve
	the weighted phase-space measure $e^{\lambda_0(x)}dv\,dx$, which is then
	taken into account by the modified weight in the corresponding
	$L^2$ space on which $e^{-\lambda_0}{\mathcal T}$ would act.
	We drop the factor $e^{-\lambda_0}$ until the
	skew-adjointness of ${\mathcal T}$ is proven and then restore it.
\end{remark}

The following integration by parts formula
shows that the transport operator ${\mathcal T}$ is skew-symmetric
with respect to weighted $L^2$ scalar products, at least when
defined on smooth functions; it should be noticed that no boundary
terms appear in this formula.

\begin{proposition} \label{transibp}
	Let $\chi \in C (]-\infty , E_0[)$.
	Then, for any $f,g \in C^1(S_0)$,
	\[
	\int_{S_0} \chi\circ E\, f\, {\mathcal T} g = - \int_{S_0} \chi\circ E\, {\mathcal T} f\, g,
	\]
	provided both integrals exist which for example is the case
	if $f \in C^1_c(S_0)$.
\end{proposition}
\begin{proof}
	We abbreviate $z=(x,v) \in {\mathbb R}^3 \times {\mathbb R}^3$ and define
	$S_n =\{z \in {\mathbb R}^6 \mid E(z) < E_0-1/n\}$.
	The properties of the characteristic flow and a change of
	variables imply that for every $t\in{\mathbb R}$ and $n\in {\mathbb N}$,
	\[
	\int_{S_n} \chi( E(z)) \,f((X,V)(t,z)) \, g((X,V)(t,z))\, dz =
	\int_{S_n} \chi( E(z)) \,  f(z) \, g(z)\, dz.
	\]
	The restriction to $S_n$ and the properties of $f, g$, and $\chi$
	guarantee that the left hand side can be differentiated with
	respect to $t$ under the integral. Thus
	\[
	\begin{split}
	0 &= \frac{d}{dt}  \int_{S_n} \chi( E(z)) \,  f((X,V)(t,z)) \,
	g((X,V)(t,z))\, dz \big|_{t=0} \\
	&= \int_{S_n} \chi( E(z)) \,
	\left( {\mathcal T} f(z)\, g(z) + f(z)\, {\mathcal T} g(z) \right)\, dz.
	\end{split}
	\]
	Hence the assertion follows with $S_n$ instead of $S_0$,
	and taking $n\to \infty$ completes the proof.
\end{proof}

A measurable function
$f \colon {\mathbb R}^3 \times {\mathbb R}^3\to {\mathbb R}$ is
\emph{spherically symmetric}, iff for every $A\in \mathrm{SO}(3)$
the identity $f(Ax,Av) = f(x,v)$ holds for a.~a.\
$(x,v)\in{\mathbb R}^3 \times {\mathbb R}^3$;
the exceptional set of measure zero may depend on $A$.
Functions defined on $S_0$ will always be extended by $0$
to ${\mathbb R}^3 \times {\mathbb R}^3$
so that spherical symmetry is also defined for them.
The subspace of spherically symmetric
functions in a certain function space is denoted by the
subscript $r$, for example $C^1_{c,r}(S_0)$ denotes the
space of compactly supported, spherically symmetric $C^1$ functions
on the set $S_0$.

We now give the definition of the transport operator which will
make it skew-adjoint.

\begin{definition} \phantomsection \label{Tweakdef}     
	\begin{itemize}
		\item[(a)]
		For $f \in L^1_{\mathrm{loc},r} (S_0)$, ${{\mathcal T} f}$ \emph{exists weakly}
		if there exists $\mu \in L^1_{\mathrm{loc},r} (S_0)$ such that
		for every test function $\xi \in C^\infty_{c,r}(S_0)$,
		\[
		\int_{S_0} \frac 1 {\vert \phi ' \circ E \vert} f\,  {\mathcal T} \xi
		= - \int_{S_0} \frac 1 {\vert \phi ' \circ E \vert} \mu\, \xi .
		\]
		In this case ${\mathcal T} f := \mu$ \emph{weakly}.
		\item[(b)]
		The real Hilbert space
		\[
		H:=\left\{f \in L^1_{\mathrm{loc},r} (S_0) \mid
		\int_{S_0} \frac 1 {\vert \phi ' \circ E \vert} f^2 < \infty\right\}
		\]
		is equipped with the scalar product
		\[
		\langle f,g\rangle_H :=
		\int_{S_0} \frac 1 {\vert \phi ' \circ E \vert} f\, g .
		\]
		\item[(c)]
		The domain of ${\mathcal T}$ is defined as
		\[
		D ({\mathcal T})
		:= \{ f \in H \mid {\mathcal T} f \ \mbox{exists weakly and}\
		{\mathcal T} f \in H \}.
		\]
	\end{itemize}  
\end{definition}

\noindent
We can now formulate our main results.

\begin{theorem}\label{skewadj}
	The transport operator ${\mathcal T} \colon H \supset D({\mathcal T}) \to H$ is skew-adjoint,
	i.e., ${\mathcal T}^\ast = -{\mathcal T}$, in both the Newtonian and the relativistic cases.
\end{theorem}

\noindent
This theorem will be proven in the next section, but some immediate comments
are in order.

\begin{remark}
	\begin{itemize}
		\item[(a)]
		If ${\mathcal T} f$ exists weakly for some $f \in L^1_{\mathrm{loc},r} (S_0)$,
		then it is uniquely determined a.~e.\ on $S_0$.
		\item[(b)]
		If $f\in C^1_{c,r}(S_0)$, then the weak and the classical definition
		\eqref{Tclassdef} of ${\mathcal T} f$ coincide due to Proposition~\ref{transibp}
		and since ${\mathcal T}$ preserves spherical symmetry.
		In particular,
		$C^1_{c,r}(S_0)\subset D({\mathcal T})$ so that
		the operator ${\mathcal T}$  in Theorem~\ref{skewadj} is densely defined.
		\item[(c)]
		If ${\mathcal T} f$ exists weakly for some $f \in L^1_{\mathrm{loc},r} (S_0)$,
		then the formula in Definition~\ref{Tweakdef}~(a) holds for
		all test functions $\xi \in C^1_{c,r}(S_0)$.
		To see this we take such a function
		and define $\xi_k := J_k \ast \xi \in C^\infty_{c,r}(S_0)$,
		where $J \in C^\infty_{c,r}(B_1(0))$ is a mollifyer with $\int J =1$,
		and $J_k := k^6 J(k\cdot)$ for $k\in {\mathbb N}$.
		There exists a compact subset $K\subset S_0$ such that
		${\mathrm{supp}\,} \xi_k \subset K$ for $k$ sufficiently large,
		and $\xi_k \to \xi$, ${\mathcal T} \xi_k \to {\mathcal T} \xi$ as $k\to\infty$,
		uniformly on $K$. This implies the identity in
		Definition~\ref{Tweakdef}~(a) for $\xi$.
		\item[(d)]    
		In the relativistic case we restricted ourselves to spherically
		symmetric functions from the start, but also in the applications in the
		Newtonian case the skew-adjointness is needed on spherically symmetric
		functions.
	\end{itemize}
\end{remark}

\begin{theorem} \label{jeans}
	In the Newtonian case,
	\[
	\mathrm{ker}\, {\mathcal T} =
	\{ f\in H \mid \exists\, g \colon {\mathbb R}^2 \to {\mathbb R} :\
	f(x,v) = g(E(x,v),L(x,v))\ \mbox{a.~e.\ on}\ S_0\}.
	\]
\end{theorem}
\noindent
This theorem will be proven in Section~\ref{seckernel}.

\begin{remark}
	\begin{itemize}
		\item[(a)]
		Since ${\mathcal T} f=0$ means that $f$ satisfies the stationary Vlasov equation
		with the potential $U_0$,
		Theorem~\ref{jeans} generalizes the fact that for spherically symmetric
		steady states of the Vlasov-Poisson system the density $f$ on phase space
		depends only on the quantities $E$ and $L$, a fact known as Jeans' Theorem,
		cf.~\cite{BFH}.
		\item[(b)]
		In the relativistic case  Jeans' Theorem is false,
		cf.~\cite{Sch}. Since in the present paper we restrict
		ourselves to isotropic steady states, Theorem~\ref{jeans} might
		still be correct also in the
		relativistic case, but the proof given below does not work there.
	\end{itemize}
\end{remark}

\section{Proof of Theorem~\ref{skewadj}---Skew-adjointness of ${\mathcal T}$}
\label{secskew}

In view of Proposition~\ref{transibp},
the main tool for the proof of Theorem~\ref{skewadj} is to approximate
a function from $D ({\mathcal T})$ by smooth functions
in such a way that the images under ${\mathcal T}$ converge
as well. As a first step we exploit the fact that
${\mathcal T} E = {\mathcal T} L = 0$.

\begin{lemma} \phantomsection \label{transweakclos}     
	\begin{itemize}
		\item[(a)]
		Let $f \in D ({\mathcal T})$ and $\chi \in C^1([0,\infty[)$ be such that
		$\chi \circ L\, f, \chi \circ L\, {\mathcal T} f\in H$.
		Then $\chi \circ L \, f\in D ({\mathcal T})$ with
		${\mathcal T} \left(\chi \circ L\, f\right) = \chi \circ L\,{\mathcal T} f$
		weakly.
		\item[(b)]
		Let $f \in D ({\mathcal T})$ and $\chi \in C(]-\infty,E_0[)$
		be such that $\chi \circ E \, f, \chi \circ E \, {\mathcal T} f\in H$. 
		Then $\chi \circ E \, f\in D ({\mathcal T})$ with
		${\mathcal T} \left(\chi \circ E \, f\right) =  \chi \circ E\, {\mathcal T} f$
		weakly.
	\end{itemize} 
\end{lemma}

\begin{proof}
	As to part (a), let $\xi \in C^1_{c,r} (S_0)$ be a test function.
	Since $\chi \circ L \in C^1_r(S_0)$, we know that
	$\chi \circ L\, \xi\in C^1_{c,r}(S_0)$ as well.
	Since $L$ is constant along characteristics, ${\mathcal T} L = 0$
	and therefore ${\mathcal T}(\chi \circ L\, \xi) = \chi \circ L\, {\mathcal T} \xi$
	classically. Thus, by Definition~\ref{Tweakdef},
	\[
	\int_{S_0} \frac 1{\vert \phi' \circ E \vert}
	f \, \chi \circ L \, {\mathcal T} \xi
	=
	\int_{S_0} \frac 1{\vert \phi' \circ E \vert}
	f \, {\mathcal T}( \chi \circ L \, \xi)
	=
	- \int_{S_0} \frac 1{\vert \phi' \circ E \vert} {\mathcal T} f
	\, \chi \circ L \,  \xi .
	\]
	The proof of part (b) is exactly the same,
	provided $\chi \in C^1(]-\infty,E_0[)$.
	If $\chi$ is only continuous we mollify $\chi$;
	we omit the details since below we need to apply the lemma with
	$\chi = |\phi'|$, and for all steady states of interest this function
	is in $C^1(]-\infty,E_0[)$.
\end{proof}

\begin{corollary} \label{transweakdefl2}
	Let $f \in  D ({\mathcal T})$. Then
	for any test function $\xi \in C^1_{c,r}(S_0)$,
	\[
	\int_{S_0} f \, {\mathcal T} \xi = - \int_{S_0} {\mathcal T} f \, \xi.
	\]
\end{corollary}
\begin{proof}
	We apply  Lemma~\ref{transweakclos} to
	$\chi = \vert \phi ' \vert$, which if necessary has to be cut
	to ensure integrability; due to the compact support of test functions
	we may restrict ourselves to a compact subset of $S_0$.
\end{proof}

In order to approximate functions in $D({\mathcal T})$ in a suitable way we
have to mollify them, and for technical reasons this needs to be done
in coordinates which are adapted to spherical symmetry. For
$(x,v) \in {\mathbb R}^3\times {\mathbb R}^3$ we define
\[
r:=|x|,\ w:=\frac{x\cdot v}{r},\ L:= |x\times v|^2;
\]
the radial velocity $w$ is defined only if $x\neq 0$.
A function $f\in H$ is spherically
symmetric in the sense defined above iff
there exists a measurable function
$f^r \colon ]0, \infty [ \times {\mathbb R} \times ]0, \infty [ \to {\mathbb R}$
such that
\[
f(x,v)  = f^r (r,w,L)\ \mbox{for a.~e.}\ (x,v)\in {\mathbb R}^3 \times {\mathbb R}^3.
\]
In this case, $f^r$ is uniquely defined a.~e.\ on
$]0, \infty [ \times {\mathbb R} \times ]0, \infty [$.
In what follows an upper index $r$ will indicate that a spherically symmetric
object is expressed in the variables $r,w,L$. In particular,
\[
E^r(r,w,L) =
\left\{
\begin{array}{ll}
\displaystyle
\frac{1}{2} w^2 + \frac{L}{2 r^2} + U_0(r),&\ \mbox{Newtonian case},\\
\displaystyle
e^{\mu_0(r)}\sqrt{1+w^2  + \frac{L}{r^2}},&\ \mbox{relativistic case},
\end{array}
\right.
\]
\[
S_0^r = \{(r,w,L)\in ]0,\infty[\times {\mathbb R} \times ]0,\infty[ \mid
E^r(r,w,L) < E_0 \},
\]
and
\[
{\mathcal T}^r f = \partial_w E^r \partial_r f - \partial_r E^r \partial_w f,
\]
for $f\in C^1_c(S_0^r)$.
The transport operator ${\mathcal T}$ preserves spherical symmetry,
and for every $\xi \in C^1_{c,r} (S_0)$ with the property
that $\xi^r \in C^1_c (S_0^r)$,
\[
{\mathcal T}^r \xi^r = \left( {\mathcal T} \xi \right)^r\ \mbox{on}\ S_0^r ;
\] 
note that $\xi \in C^1_{c,r} (S_0)$ does not imply that
$\xi^r \in C^1_c (S_0^r)$, since the support of functions
in $C^1_c (S_0^r)$ has to be bounded away from  $r=0$ and $L = 0$.
The integration-by-parts formula in Proposition~\ref{transibp} can be
rewritten as
\begin{equation}
\begin{split}\label{intpartrwL}
\int_{S_0^r}
& \chi(E^r(r,w,L)) \,  f (r,w,L) \, ({\mathcal T}^r g )(r,w,L)
\,dr\,dw\,dL \\
& = -
\int_{S_0^r} \chi(E^r(r,w,L)) \,  ( {\mathcal T}^r f ) (r,w,L) \, g(r,w,L)\,dr\,dw\,dL
\end{split}
\end{equation}
for all $f,g \in C^1_c(S_0^r)$ and $\chi \in C (]-\infty , E_0[)$;
this follows from the spherical symmetry of the integrands and the
change-of-variables formula, observing that $dv = \frac{\pi}{r^2} dL\, dw$
and $dx = 4 \pi r^2 dr$.
Similarly, Corollary~\ref{transweakdefl2} can be rewritten in the
variables $r,w,L$, i.e., \eqref{intpartrwL} remains valid with
$\chi=1$, $g \in C^1_c(S_0^r)$, and $f$ replaced by the representative
$f^r$ of a function $f\in D({\mathcal T})$.

We can now prove the desired approximation result.

\begin{proposition} \label{transapprox}
	Let $f \in D ({\mathcal T})$. Then there exists a sequence
	$(F_k)_{k\in{\mathbb N}} \subset C^{\infty}_{c,r} (S_0)$ such that
	$F_k^r \in C^{\infty}_c (S_0^r)$ for $k\in{\mathbb N}$ and
	\[
	F_k \to f \ \mbox{and}\ {\mathcal T} F_k \to {\mathcal T} f\ \mbox{in}\  H  \ \mbox{as}\
	k\to\infty .
	\]
\end{proposition}
\begin{proof}
	We split the proof into several steps.\\
	\textit{Step1: Reduction to a compact support.}
	For each $k\in{\mathbb N}$ let $\chi_k \in C^{\infty} ({\mathbb R})$ be an
	increasing cut-off function such that 
	\[
	\chi_k (s) = 0 \ \mbox{for}\  s \leq \frac 1{2k},
	\qquad \chi_k (s) = 1\ \mbox{for}\  s \geq \frac 1k .
	\]
	For $(x,v) \in S_0$ and $k\in{\mathbb N}$ let
	\[
	f_k(x,v) := {\chi}_k ( L(x,v)) \, {\chi}_k (E_0 - E(x,v)) \, f(x,v).
	\]
	The boundedness of $ \chi_k$ together with the spherical symmetry of
	$E$ and $L$ ensure that $f_k \in H$. Hence by Lemma~\ref{transweakclos},
	$f_k \in D ({\mathcal T})$ with
	${\mathcal T} f_k = (\chi_k \circ L) \, ( \chi_k \circ (E_0 - E) ) \, ({\mathcal T} f)$,
	and by Lebesgue's dominated convergence theorem,
	\[
	f_k \to f\ \mbox{and}\ {\mathcal T} f_k \to {\mathcal T} f\ \mbox{in}\ H
	\ \mbox{as}\ k\to\infty.
	\]
	By applying the following arguments to $f_k$ for $k\in{\mathbb N}$ sufficiently large
	instead of to $f$, we may assume that $f^r$ has compact support in $S_0^r$
	and that there exists $m\in{\mathbb N}$ such that for a.~e.\
	$(r,w,L) \in S_0^r$ with $f^r(r,w,L) \neq 0$,
	\[
	E^r(r,w,L) < E_0 - \frac 1m < 0\ \mbox{and}\
	\bar B_{\frac 1m} (r,w,L) \subset S_0^r .
	\]
	This also implies that ${\mathcal T} f = 0$ a.~e.\ on
	$\{(x,v) \in S_0 \mid \vert x \vert \leq \frac 1m
	\lor L(x,v) \leq \frac 1m \lor E(x,v) \geq E_0 - \frac 1m \}$.
	Furthermore, $( {\mathcal T} f)^r$ has compact support in $S_0^r$.
	
	\noindent
	\textit{Step 2: The approximating sequence.}
	We first introduce some terminology. We let
	\[
	H^r
	:= \{ f \colon S_0^r \to {\mathbb R} \ \mbox{measurable} \mid
	\| f \|_{H^r} < \infty \} ,
	\]
	where the norm $\| \cdot \|_{H^r}$ is induced by the scalar product
	\[
	\langle f, g \rangle_{H^r} :=
	4 \pi^2 \int_{S_0^r} \frac{1}{\vert \phi' (E^r(r,w,L)) \vert}
	f(r,w,L) \, g(r,w,L)
	\,dr\,dw\,dL .
	\]
	The space $L^2 (S_0^r)$ is obtained by dropping the weight
	$1/|\phi'\circ E^r|$, but the factor $4 \pi^2$ is included
	in the corresponding scalar product and norm. Hence
	the change-of-variables formula shows that
	the map $f\mapsto f^r$ is an isometric isomorphism of
	the spaces $L^2_{r} (S_0) \cong L^2 (S_0^r)$ or $H\cong H^r$.
	
	Let $J \in C^{\infty}_{c} (B_1(0) )$ be a three-dimensional mollifyer, i.e.,
	$B_1(0)\subset {\mathbb R}^3$, $\int_{{\mathbb R}^3} J =1$, and $J \geq 0$. We define
	$J_k := k^3 J(k\cdot )$ for $k\in{\mathbb N}$.
	Due to the compact support of $f^r$ and $({\mathcal T} f)^r$ in $S_0^r$,
	standard mollifying arguments, a change of variables, and the fact that
	the weight $1/|\phi'\circ E|$ is a strictly positive,
	bounded, continuous function on any compact subset of $S_0^r$ it follows that
	\[
	f^r \in L^2 (S_0^r)\ \mbox{and}\
	J_k \ast f^r \to f^r\ \mbox{in}\ L^2(S_0^r)\ \mbox{and in}\
	H^r\ \mbox{as}\ k\to\infty,
	\]
	\[
	({\mathcal T} f)^r \in L^2 (S_0^r)\  \mbox{and}\
	J_k \ast ({\mathcal T} f)^r \to ({\mathcal T} f)^r\ \mbox{in}\ L^2(S_0^r)\ \mbox{and in}\
	H^r\ \mbox{as}\ k\to\infty .
	\]
	\textit{Step 3: Boundedness.} 
	We want to prove that
	$({\mathcal T} ( J_k \ast f^r) )_{k\in{\mathbb N}} \subset L^2(S_0^r)$ is bounded,
	from which we can then obtain the weak convergence of a subsequence.
	The properties of the support of $f^r$, the boundedness of $S_0^r$,
	and the mean value theorem yield the existence of a constant
	$C_E \geq 1$ such that 
	\[
	\vert E^r(z) - E^r(z') \vert \leq C_E \vert z- z' \vert\ \mbox{for}\
	z,z' \in \bar B_{\frac 1m} ({\mathrm{supp}\,} f^r) \subset S_0^r ;
	\]
	$B_\epsilon(M) := \cup_{x\in M} B_\epsilon (x)$ for $\epsilon >0$ and some set
	$M\subset {\mathbb R}^3$.
	We choose $k\in{\mathbb N}$ such that $k >2 C_E m$.
	For every $z=(r,w,L) \in S_0^r$ with $(J_k \ast f^r)(z) \neq 0$,
	it then follows that $\bar B_{1/(2m)} (z) \subset S_0^r$
	and $E^r(z) < E_0 -1/(2m)$.
	In particular, $J_k \ast f^r \in C^{\infty}_c (S_0^r)$ for $k >2 C_E m$.
	For these $k$ and $z \in {\mathrm{supp}\,}(J_k\ast f^r)$ 
	it follows that
	\[
	\begin{split}
	&
	\left[{\mathcal T}^r ( J_k \ast f^r)\right] (z)
	= \partial_w E^r(z) \, \left[ (\partial_r J_k) \ast f^r \right](z)
	- \partial_r E^r(z) \, \left[ (\partial_w J_k) \ast f^r \right] (z) \\
	&
	= \int_{S_0^r} \left[\partial_w E^r(z)  \, \partial_r J_k (z-z') -
	\partial_r E^r(z) \,\partial_w J_k (z-z') \right] f^r (z') \,d z'\\
	&
	= \int_{B_{1/k}(z)} \Bigl[(\partial_w E^r(z) -\partial_w E^r(z')) \,
	\partial_r J_k (z-z')\\
	&\qquad\qquad\qquad {}-
	(\partial_r E^r(z) -\partial_r E^r(z'))\, \partial_w J_k (z-z') \Bigr]
	f^r (z')\,d z' \\
	&\qquad
	{}+ \int_{S_0^r} \left[\partial_w E^r(z') \, \partial_r J_k (z-z')
	- \partial_r E^r(z') \, \partial_w J_k (z-z') \right]\, f^r(z') \,dz',
	\end{split}
	\]
	where $z' = (r', w', L')$. Since $\partial_r E^r$ and $\partial_w E^r$
	are Lipschitz on $\bar B_{1/m} ({\mathrm{supp}\,} f^r )$, the absolute value of
	the first term can be estimated by
	\[
	\begin{split}
	&
	\frac Ck \int_{B_{1/k}(z)} \vert D J_k (z-z') \vert \,
	\vert f^r (z') \vert \,dz' 
	= \frac Ck \int_{B_{1/k}(0)} \vert D J_k (\tilde z) \vert \,
	\vert f^r(z- \tilde z) \vert \,d \tilde z \\
	&
	= C k^3 \int_{B_{1/k}(0)} \vert D J (k \tilde z) \vert \,
	\vert f^r(z- \tilde z) \vert \,d \tilde z
	= C k^3 \left( \vert D J (k \cdot) \vert \ast
	\vert f^r \vert \right) (z),
	\end{split}
	\]
	where $C > 0$ depends on the support of $f$ and the fixed steady state $f_0$
	and $DJ_k$ denotes the total derivative of $J_k$.
	For the second term we note that
	$J_k (z - \cdot) \in C^1_c (S_0^r)$ for $k > 2 C_E m$,
	since ${\mathrm{supp}\,} J_k (z - \cdot) \subset \bar B_{1/(2m)} (z) \subset S_0^r$.
	By \eqref{intpartrwL} and the comment after that equation we conclude that
	\[
	\begin{split}
	\int_{S_0^r}
	&
	\left[ \partial_w E(z') \, \partial_r J_k (z-z') -
	\partial_r E(z') \, \partial_w J_k (z-z') \right]\, f^r(z') \,dz'\\
	&
	= - \langle {\mathcal T}^r \left[ J_k (z - \cdot )\right], f^r \rangle_2
	= \langle J_k (z - \cdot ), ({\mathcal T} f)^r \rangle_2
	= \left[ J_k \ast ({\mathcal T} f)^r \right] (z) .
	\end{split}
	\]
	Altogether, we get the estimate
	\[
	\begin{split}
	\| {\mathcal T} (J_k \ast f^r) \|_2
	&\leq
	\|J_k \ast ({\mathcal T} f)^r \|_2
	+ C k^3 \| \vert D J (k \cdot) \vert \ast \vert f^r \vert \|_2 \\
	&
	\leq \|J_k \ast ({\mathcal T} f)^r \|_2 + C k^3 \| f^r \|_2 \, \|D J (k \cdot) \|_1,
	\end{split}
	\]
	where we used Young's inequality.
	Since $J_k \ast ({\mathcal T} f)^r \to ({\mathcal T} f)^r$ in $L^2(S_0^r)$ as $k\to\infty$,
	the first term is bounded. As to the second, we note that
	$\| D J (k \cdot) \|_1 = k^{-3} \| D J \|_1$, and the 
	desired boundedness is proven.
	
	\noindent
	\textit{Step 4: Weak convergence.}
	Let $f_k^r := J_{k} \ast f^r$.
	Due to the previous step there exists a subsequence
	which by abuse of notation we again denote as $(f^r_k)_{k\in{\mathbb N}}$
	and a limit $g^r \in L^2(S_0^r)$ such that 
	\[
	{\mathcal T}^r f^r_k \rightharpoonup g^r\ \mbox{in}\ L^2(S_0^r)\ \mbox{as}\
	k\to\infty.
	\]
	We need to show that $g = {\mathcal T} f$, where
	for $(x,v) \in S_0$ with $x\times v \neq 0$,
	\[
	g(x,v) :=
	g^r \left(\vert x \vert, \frac{x\cdot v}{|x|} , | x \times v |^2\right).
	\] 
	Let $\xi \in C^1_{c,r} (S_0)$ be an arbitrary test function.
	We have to ensure that $\xi^r \in C^1_c (S_0^r)$,
	which can be achieved due to the compact support of $f^r$ in $S_0^r$.
	From the properties of the support of $f^r_k$ shown in Step~3,
	\[
	(\chi_{2m} \circ L) \, (\chi_{2m} \circ (E_0 - E) ) \, {\mathcal T}^r f^r_k
	= {\mathcal T}^r f^r_k
	\]
	if $k$ is sufficiently large.
	Let $\tilde \xi := (\chi_{2m} \circ L) \, (\chi_{2m} \circ (E_0 - E) ) \, \xi$
	and note that $\tilde \xi ^r \in C^1_c (S_0^r)$. In addition,
	$f \, {\mathcal T} \tilde \xi = f \, {\mathcal T}  \xi$ and $g\, \tilde \xi = g\, \xi$
	a.~e.\ on $S_0$, where the latter follows from the properties of the
	support of $f^r_k$ and the Du~Bois-Reymond theorem.
	Thus, changing variables yields
	\[
	\begin{split}
	\langle g , \xi \rangle_H
	&
	= \langle g ,\tilde \xi \rangle_H
	= \langle g^r ,\tilde \xi^r \rangle_{H^r}
	= \lim_{k\to \infty} \langle {\mathcal T}^r f^r_k, \tilde \xi^r \rangle_{H^r}
	= - \lim_{k\to \infty} \langle f^r_k, {\mathcal T}^r \tilde \xi^r \rangle_{H^r}\\
	&
	= - \langle f^r, ({\mathcal T} \tilde \xi)^r \rangle_{H^r} 
	= - \langle f, {\mathcal T} \tilde \xi \rangle_H 
	= - \langle f, {\mathcal T} \xi \rangle_H ,
	\end{split}
	\]
	where we used \eqref{intpartrwL}
	and the fact that due to the compact support of $\tilde \xi^r$,
	$\langle \cdot , \tilde \xi^r \rangle_{H^r}\in L^2 (S_0^r)'$.
	
	\noindent
	\textit{Step 5: Strong convergence.} 
	By the previous step, ${\mathcal T}^r f_k^r \rightharpoonup ({\mathcal T} f)^r$ in
	$L^2 (S_0^r)$ as $k\to\infty$. Mazur's lemma implies that for
	every $k\in{\mathbb N}$ there exists $N_k \geq k$ and weights
	$c_k^k , \ldots , c_{N_k}^k \in [0,1]$ with $\sum_{j=k}^{N_k} c_j^k = 1$
	such that 
	\[
	{\mathcal T}^r \left( \sum_{j=k}^{N_k} c_j^k f_j^r \right)
	= \sum_{j=k}^{N_k} c_j^k {\mathcal T}^r f_j^r \to ({\mathcal T} f)^r \ \mbox{in}\
	L^2 (S_0^r) \ \mbox{as}\ k\to\infty .
	\]
	Let $F_k^r := \sum_{j=k}^{N_k} c_j^k f_j^r$ for $k\in{\mathbb N}$. Then
	$F_k^r \to f^r$ in $L^2(S_0^r)$ and $H^r$ as $k\to\infty$.
	Also $F_k^r \in C^{\infty}_{c} (S_0^r)$ for $k$ sufficiently large.
	Finally, 
	\[
	F_k (x,v) :=
	F_k^r \left(| x |, \frac{x \cdot v}{| x |}, | x \times v |^2\right)\
	\mbox{for}\
	(x,v) \in S_0\ \mbox{with}\ x \times v \neq 0,
	\]
	extended by $0$ on $S_0$, defines a function
	$F_k \in C^{\infty}_{c,r} (S_0)$ for $k$ sufficiently large,
	due to the compact support of $F_k^r$. Changing variables yields 
	\[
	F_k \to f\ \mbox{in}\ H \ \mbox{and}\ {\mathcal T} F_k \to {\mathcal T} f\ \mbox{in}\
	L^2 (S_0)\ \mbox{as}\ k\to\infty,
	\]
	and the support properties of the involved functions
	allow us to conclude that ${\mathcal T} F_k \to {\mathcal T} f$ in $H$ as well,
	which finishes the proof. 
\end{proof}

Let us recall the definition of the adjoint operator
${\mathcal T}^* \colon H \supset D ({\mathcal T}^*) \to H$ of the operator ${\mathcal T}$.
Its domain of definition is
\[
D ({\mathcal T}^*) := 
\{ f \in H \mid \exists_1 h \in H\, \forall g \in D ({\mathcal T}) \colon
\langle {\mathcal T} g , f \rangle_H = \langle  g , h \rangle_H \}.
\]
For $f \in D ({\mathcal T}^*)$, ${\mathcal T}^* f := h$; note that ${\mathcal T}$ is a densely
defined operator on the Hilbert space $H$.

\noindent
\begin{proof}[Proof of Theorem~\ref{skewadj}]
	First we observe that ${\mathcal T}$ is skew-symmetric, i.e.,
	for all $f,g \in D({\mathcal T})$ it holds that 
	\begin{equation} \label{transskewsymm}
	\langle f , {\mathcal T} g \rangle_H = - \langle {\mathcal T} f , g \rangle_H.
	\end{equation}
	This follows by approximating one of the two functions
	via Proposition~\ref{transapprox} and then using
	Definition~\ref{Tweakdef}.
	
	Now let $f\in D ({\mathcal T})$. Then \eqref{transskewsymm}
	implies that
	\[
	\langle {\mathcal T} g, f \rangle_H = - \langle g, {\mathcal T} f  \rangle_H
	= \langle g, - {\mathcal T} f \rangle_H
	\]
	for all $g \in D ({\mathcal T})$, i.e.,
	$f \in D ({\mathcal T}^*)$ with ${\mathcal T}^* f = - {\mathcal T} f$,
	and hence $-{\mathcal T} \subset {\mathcal T}^\ast$.
	
	If on the other hand  $f\in D ({\mathcal T}^*)$ and $h \in H$ are such that
	$\langle {\mathcal T} g , f \rangle_H = \langle  g , h \rangle_H$ for all
	$g\in D ({\mathcal T})$, then since $C^1_{c,r} (S_0) \subset D({\mathcal T})$ this implies that 
	\[
	\langle f , {\mathcal T} \xi \rangle_H = - \langle h , \xi \rangle_H
	\]
	for all test functions $\xi \in C^1_{c,r} (S_0)$.
	By Definition~\ref{Tweakdef}, this means that
	$f\in D({\mathcal T})$ with ${\mathcal T} f = - h$, i.e., ${\mathcal T}^\ast \subset -{\mathcal T}$,
	and the proof is complete. \end{proof}

We conclude this section with some further remarks.

\begin{remark}
	\begin{itemize}
		\item[(a)]
		As noted in the introduction the relevant operator in the general
		relativistic case is actually given as $\widetilde{{\mathcal T}} = e^{-\lambda_0} {\mathcal T}$.
		Let us denote by $\widetilde{H}$ the Hilbert space $H$ equipped with
		the scalar product
		\[
		\langle f,g\rangle_{\widetilde{H}} :=
		\int_{S_0} \frac {e^{\lambda_0}}{\vert \phi ' \circ E \vert} f\, g .
		\]
		Then the transport operator $\widetilde{{\mathcal T}} \colon
		\widetilde{H} \supset D({\mathcal T}) \to \widetilde{H}$ is skew-adjoint;
		note that $D(\widetilde{{\mathcal T}})=D({\mathcal T})$.
		\item[(b)]
		Theorem~\ref{skewadj} remains correct without the assumption
		of spherical symmetry, i.e., if this symmetry assumption is
		removed from all the relevant function spaces. Indeed, the proof
		then becomes simpler, but as mentioned before, in the applications
		the operator ${\mathcal T}$ is defined on spherically symmetric functions,
		and Theorem~\ref{skewadj} including this assumption is needed
		in the proof of Theorem~\ref{jeans}.
		\item[(c)]
		The relation \eqref{Tinfgen} to the characteristic flow induced by
		\eqref{charsys} suggests an alternative route to study the operator
		${\mathcal T}$. Since the characteristic flow is measure preserving,
		\[
		(U (s) f) (x,v) := f( X(s,x,v), V(s,x,v)) , \quad (x,v) \in S_0,\
		s\in {\mathbb R},
		\]
		defines a unitary $C^0$-group on $H$.
		By Stone's theorem, this $C^0$-group has a unique skew-adjoint
		infinitesimal generator $\tilde {\mathcal T}$ defined on the dense subset
		\[
		D (\tilde {\mathcal T}) := \left\{ f \in H \mid \lim_{s\to 0}
		\frac{U (s) f - f}{s}\ \mbox{exists in}\ H \right\}
		\]
		by 
		\[
		\tilde {\mathcal T} f := \lim_{s\to 0} \frac{U (s) f - f}{s}.
		\]
		Since ${\mathcal T}$ is skew-adjoint on $H$ as well and
		${\mathcal T}$ and $\tilde {\mathcal T}$ coincide on the dense subset $C^1_{c,r} (\Omega_0)$,
		${\mathcal T} = \tilde {\mathcal T}$, in particular $D ({\mathcal T}) = D (\tilde {\mathcal T})$.
	\end{itemize}
\end{remark}

\section{Proof of Theorem~\ref{jeans}---The kernel of ${\mathcal T}$}
\label{seckernel}

We begin this section by stressing that from now on
all the arguments refer only to the Newtonian case, i.e.,
to the Vlasov-Poisson system. In that context
a smooth, spherically symmetric
solution $f$ of the equation ${\mathcal T} f=0$ depends
only on the quantities $E$ and $L$, cf.~\cite{BFH}.
Functions in $D({\mathcal T})$ need not be smooth, and it therefore seems natural
to prove Theorem~\ref{jeans} by mollifying such functions like in the
previous section. 
The mollification of a function in the kernel of ${\mathcal T}$ need
not belong to the kernel anymore,
but we will show that the distance between elements of an
approximating sequence obtained by mollification
and their projection onto the space of functions depending only on $(E,L)$
tends to zero.
In order to define this projection, we first analyze the solutions
of the characteristic system \eqref{charsys}
in the coordinates $(r,w,L)$.
Since $L$ is constant along characteristics, we treat
it as a parameter. For fixed $L>0$ the particle trajectories
are governed by
\begin{equation} \label{csrwL}
\dot{r} = w , \; \dot{w} = - \psi_L'(r),
\end{equation}
where the effective potential $\psi_L$ is defined as
\[
\psi_L(r) := U_0(r) + \frac{L}{2 r^2} .
\]
We need the following properties of this effective potential;
these results can be found in \cite{GuRe2007,LeMeRa1}, but  for the sake of
completeness, we include their proofs.

\begin{lemma} \phantomsection \label{effpot}
	\begin{itemize}
		\item[(a)]
		For any $L>0$ there exists a unique $r_L > 0$ such that 
		\[
		\min_{\left]0, \infty \right[}( \psi_L ) = \psi_L(r_L) < 0.
		\]
		The mapping
		$\left]0, \infty \right[ \ni L \mapsto r_L$
		is continuously differentiable.
		\item[(b)]
		For any $L>0$ and $E \in \left] \psi_L (r_L) , 0 \right[$
		there exist two unique radii
		\[
		0 < r_-(E,L) < r_L < r_+(E,L) < \infty
		\]
		such that $\psi_L (r_\pm(E,L)) = E $. The functions
		\[
		\{ (E,L) \in \left]- \infty, 0 \right[ \times \left]0, \infty \right[
		\mid \psi_L(r_L) < E \} \ni (E,L) \mapsto r_\pm(E,L)
		\]
		are continuously differentiable. 
		\item[(c)]
		For any $L>0$ and $E \in \left] \psi_L (r_L) , 0 \right[$,
		\[
		r_+(E,L) < - \frac{M_0}{E},
		\]
		where $M_0:=||f_0||_1\in \,]0,\infty[$
		denotes the total mass of the steady state.
		\item[(d)]
		For any $L>0$,
		$E \in \left] \psi_L (r_L) , 0 \right[$ and
		$r \in \left[ r_-(E,L) , r_+(E,L) \right]$
		the following estimate holds:
		\[
		E - \psi_L(r) \geq L \, \frac{(r_+(E,L) - r) \,
			(r - r_-(E,L))}{2 r^2 r_-(E,L)  r_+(E,L)}.
		\]
	\end{itemize}
\end{lemma}
\begin{proof}
	First we note the
	$\psi_L'(r) = U_0'(r) - L/r^3 = r^{-2} (m_0(r) - L/r)$,
	where $m_0(r):=4\pi \int_0^r s^2 \rho_0(s)\, ds$ is the mass within the ball
	of radius $r$ for the given steady state.
	Hence $\psi_L'(r) = 0$ is equivalent to $m_0(r) = L/r$.
	Since the mapping
	$\left] 0, \infty \right[ \ni r \mapsto m_0(r) - L/r$
	is strictly increasing and 
	\[
	\lim_{r\to 0} \left(  m_0(r) - \frac{L}{r} \right)
	= - \infty ,\; \lim_{r\to\infty} \left(  m_0(r) - \frac{L}{r} \right) > 0,
	\]
	there exists a unique radius $r_L > 0$ with $\psi_L'(r_L) = 0$
	as well as $\psi_L'(r) < 0$ for $0 < r < r_L$ and
	$\psi_L'(r) > 0$ for $r > r_L$.
	This monotonicity together with $\lim_{r\to 0} \psi_L(r) = \infty$ and
	$\lim_{r\to \infty} \psi_L(r) = \lim_{r\to \infty} U_0(r) = 0$ implies that
	$\psi_L (r_L)$ is negative and the minimal value of
	$\psi_L$ on $\left] 0, \infty \right[$.
	Since
	\[
	\frac{d}{d r} \left( m_0(r) - \frac{L}{r} \right)
	= 4 \pi r^2 \rho_0(r) + \frac L{r^2} > 0
	\]
	for all $r > 0$, the continuous differentiability follows
	by the implicit function theorem. This proves part (a),
	and we note that $r_L = L/m_0(r_L)$.
	
	As to part (b), the monotonicity of $\psi_L$
	together with its limit behavior as $r\to 0$ and $r\to\infty$
	yields the existence and uniqueness of $r_\pm(E,L)$. Since
	$\psi_L'(r) \neq 0$ for $r \neq r_L$, the implicit function theorem
	implies that the mapping $(E,L) \mapsto r_\pm (E,L)$ is
	continuously differentiable.
	
	In order to show (c) we first note that for all $r>0$,
	\[
	\begin{split}
	U_0(r) &= - \frac{m_0(r)}{r} - 4 \pi \int_r^\infty s \rho_0(s) ds \\
	&\geq - \frac{1}{r} \left( m_0(r) +
	4 \pi \int_r^\infty s^2 \rho_0(s)\ ds \right) = -\frac{M_0}{r} .
	\end{split}
	\]
	Hence every $r>0$ with $E - \psi_L (r) > 0$ also satisfies
	\[
	E + \frac {M_0} r - \frac L{2 r^2} > 0.
	\]
	This quadratic inequality implies that
	\[
	r_+(E,L) \leq \frac{L}{M_0 - \sqrt{M_0^2 + 2 EL}}
	= \frac{-M_0 - \sqrt{M_0^2 + 2EL}}{2E} < - \frac{M_0}{E};
	\]
	note that for $0>E>\psi_L(r_L)$,
	\[
	M_0^2 + 2 E L
	> M_0^2 - 2 L \frac{M_0}{r_L} + \frac{L^2}{r_L^2}
	= (M_0 - m_0(r_L))^2 \geq 0.
	\]
	As to part (d) we define for
	$r \in \left[ r_-(E,L) , r_+(E,L) \right]$,
	\[
	\xi (r) := E - \psi_L(r) - L \, \frac{(r_+(E,L) - r) \,
		(r - r_-(E,L))}{2 r^2 r_-(E,L)  r_+(E,L)}.
	\]
	Then the radial Poisson equation implies that
	\[
	\frac{d^2}{dr^2} \left[ r \xi(r) \right]
	= - \frac{1}{r} \frac{d}{dr} \left[ r^2 U_0' (r) \right]
	= - 4 \pi r \rho_0(r) \leq 0.
	\]
	The mapping
	$\left[ r_-(E,L) , r_+(E,L) \right] \ni r \mapsto r \xi(r) \in {\mathbb R}$
	is therefore concave with $\xi(r_\pm(E,L)) = 0$.
	Hence $\xi\geq 0$ on the interval
	$\left[ r_-(E,L) , r_+(E,L) \right]$, and proof is complete. 
\end{proof}
We now consider an arbitrary, global solution
${\mathbb R} \ni t \mapsto (r(t), w(t), L)$ of the characteristic system \eqref{csrwL}.
Since the particle energy is conserved along characteristics,
$E = E(r(t), w(t), L)$ for all $t\in{\mathbb R}$. We assume that
$L >0$ and $E < 0$, since other solutions are of no interest.
For any $t\in{\mathbb R}$,
\[
\psi_L (r_L) \leq \psi_L (r(t)) \leq \frac 12 w^2 (t) + \psi_L (r(t)) = E
\]  
and thus $r_-(E,L) \leq r(t) \leq r_+ (E,L) $ by Lemma~\ref{effpot}.
Solving for $w$ yields
\[
\dot r (t) = w(t)  = \pm \sqrt{ 2 E - 2 \psi_L (r(t)) }
\]
for $t\in{\mathbb R}$.  Therefore,
$r$ oscillates between $r_-(E,L)$ and $r_+(E,L)$,
where $\dot{r} = 0$ is equivalent to $r=r_\pm (E,L)$,
and $\dot{r}$ switches its sign when reaching $r_\pm(E,L)$.
By applying the inverse function theorem and integrating,
we obtain the following formula for the period of the $r$-motion, i.e.,
the time needed for $r$ to travel from $r_-(E,L)$ to $r_+(E,L)$ and back.
For $L > 0$ and $\psi_L (r_L) < E < E_0$ let
\begin{equation} \label{Tdef}
T(E,L) := 2 \int_{r_-(E,L)}^{r_+(E,L)} \frac{\,d r}{\sqrt{2E - 2\psi_L(r)}}.
\end{equation}
Since $E - \psi_L(r) > 0$ for $r_-(E,L) < r < r_+(E,L)$, this
expression is well defined.
Lemma~\ref{effpot} implies that $T(E,L)$ is finite:
\begin{equation}
\begin{split} \label{Test}
T(E,L)
& =
\sqrt{2} \int_{r_-(E,L)}^{r_+(E,L)} \frac{dr}{\sqrt{E - \psi_L(r)}}\\
& \leq
2 \int_{r_-(E,L)}^{r_+(E,L)}
\frac{r \sqrt{r_-(E,L) r_+(E,L)}}{\sqrt{L} \sqrt{(r_+(E,L)-r) \,
		(r - r_-(E,L))}} \,dr \\
& \leq
2 \frac{r_+^2(E,L)}{\sqrt{L}}  \int_0^1 \frac{ds}{\sqrt{s(1-s)}}
= 2 \pi \frac{r_+^2(E,L)}{\sqrt{L}} \leq 2 \pi \frac{M_0^2}{E^2 \sqrt{L}}.
\end{split}
\end{equation}
The projection onto the space of functions depending only on $E$ and $L$
is obtained by averaging over trajectories fixed by $(E,L)$.
For fixed $(r,w,L) \in ]0,\infty[ \times {\mathbb R} \times ]0,\infty[$ let
${\mathbb R} \ni t \mapsto (R,W)(t,r,w,L)$ be the unique global solution of
the characteristic system \eqref{csrwL}
satisfying the initial condition $(R,W)(0,r,w,L) = (r,w)$.
For $f \in H$ (extended by $0$ to ${\mathbb R}^3 \times {\mathbb R}^3$)
and $L >0$, $\psi_L (r_L) < E < E_0$ we define its projection as
\begin{equation}
\begin{split}\label{Pfdef}
{\mathcal P} & f (E,L)
:= \int_0^1 f^r ( (R,W)(t \, T(E,L), r_-(E,L), 0, L) , L ) \,dt \\
&= \frac 1{T(E,L)} \int_{r_-(E,L)}^{r_+(E,L)}
\frac{ f^r(r , \sqrt{2E - 2\psi_L (r)}, L)
	+  f^r(r , - \sqrt{\ldots}, L) }{\sqrt{2E - 2\psi_L (r)}} \,dr .
\end{split}
\end{equation}
Then ${\mathcal P} f (E,L)$ is uniquely determined for a.~e.\
$(E,L) \in {\mathbb R}^2$ satisfying $L>0$ and $\psi_L (r_L) < E < E_0$,
since 
\[
\int_{S_0}  f(x,v) \,dx\,dv
= 4 \pi^2 \int_0^{\infty} \int_{\psi_L(r_L)}^{E_0} T(E,L)\, {\mathcal P} f (E,L)\,dE \,dL.
\]
To obtain this identity we first change to the integration variables $(r,w,L)$
as explained in the derivation of \eqref{intpartrwL}, then change the order of
integration
from $(r,w,L)$ to $(L,w,r)$,
change the integration variable $w$ to $E$, and finally we express the resulting
integrand using \eqref{Pfdef}.
We want to interpret ${\mathcal P}$ as a map from $H$ into $H$, so by slight abuse
of notation we denote the function
\[
S_0 \ni (x,v) \mapsto {\mathcal P} f (E(x,v), L(x,v))
\]
also by ${\mathcal P} f$.
Then for all $f,g \in H$,
\[
\langle{\mathcal P} f,g\rangle_H
=
4 \pi^2 \int_0^{\infty} \int_{\psi_L(r_L)}^{E_0}
\frac{T(E,L)}{\vert \phi' (E) \vert} {\mathcal P} f (E,L) \,
{\mathcal P} g (E,L) \,d E \,d L
=\langle f,{\mathcal P} g\rangle_H,
\]
in particular, $\langle{\mathcal P} f,f\rangle_H = ||{\mathcal P} f||_H^2$, and hence
$||{\mathcal P} f||^2_H \leq  ||f||_H^2$ which means that ${\mathcal P}$ maps $H$ into itself.
Since ${\mathcal P} {\mathcal P} f = {\mathcal P} f$, ${\mathcal P}$
is the orthogonal projection onto the closed subspace
\[
\begin{split}
K_{{\mathcal T}}
:=&
\{ f\in H \mid  \exists\, g \colon {\mathbb R}^2 \to {\mathbb R}:\ f(x,v) = g(E(x,v), L(x,v))
\ \mbox{a.~e.\ on}\ S_0  \}\\
=&
\{ f\in H \mid f(x,v)={\mathcal P} f(E(x,v),L(x,v)) \ \mbox{a.~e.\ on}\ S_0 \}.
\end{split}
\]
Projection operators similar to ${\mathcal P}$ have for example been used
in \cite{bost,GuLi,LeMeRa1}.

\noindent
\begin{proof}[Proof of Theorem~\ref{jeans}, i.e., $\mathrm{ker}\,{\mathcal T} = K_{{\mathcal T}}$]
	We first show the easy inclusion.
	Let $f \in K_{{\mathcal T}}$, i.e., for some $g \colon {\mathbb R}^2 \to {\mathbb R}$,
	$f(z) = g(E(z), L(z))$ a.~e.\ on  $S_0$; we abbreviate $z=(x,v)$.
	Since $E$ and $L$ are conserved along characteristics, we find that
	for every $\xi \in C^1_c (S_0)$,
	\[
	\int_{S_0} \frac 1 {\vert \phi' (E(z)) \vert} f(z) \,
	\xi ( (X,V)(t,z) ) \,dz 
	= \int_{S_0} \frac 1 {\vert \phi' (E(z)) \vert} f(z) \,
	\xi (z) \,dz .
	\]
	Thus
	\[
	0 
	=\frac{d}{dt}  \left[ \int_{S_0} \frac 1 {\vert \phi' (E(z)) \vert}
	f(z) \, \xi ((X,V)(t,z)) \,dz \right] \big|_{t=0} 
	=
	\int_{S_0} \frac 1 {\vert \phi' (E(z)) \vert} f(z) \,
	\, {\mathcal T} \xi (z) \,dz.
	\]
	By Definition~\ref{Tweakdef}, this means that
	$f\in D ({\mathcal T})$ with ${\mathcal T} f = 0$, i.e., $f \in \ker \, {\mathcal T}$.
	
	As to the other inclusion, let $f \in \ker\,{\mathcal T}$, i.e.,
	$f\in D ({\mathcal T})$ with ${\mathcal T} f = 0$. As stated above, we will show $f\in K_{{\mathcal T}}$
	by approximation. We will split this argument into several steps.
	
	\noindent
	\textit{Step 1: Reduction to a compact support.}
	As before, let $\chi_k \in C^{\infty} ({\mathbb R})$ be a smooth,
	increasing cut-off function with
	\[
	\chi_k (s) = 0\ \mbox{for}\  s \leq \frac 1{2k}, \qquad
	\chi_k (s) = 1\ \mbox{for}\ s \geq \frac 1k 
	\]
	for each $k\in{\mathbb N}$. Now set
	\[
	f_k(z) := {\chi}_k ( L(z)) \, {\chi}_k (E_0 - E(z)) \, f(z)
	\]
	for $z \in S_0$ and $k\in{\mathbb N}$. Since $f_k \to f$ in $H$ as
	$k\to\infty$ and since $K_{{\mathcal T}}$ is closed, it suffices to show
	$f_k \in K_{{\mathcal T}}$ for every $k\in{\mathbb N}$ to conclude $f\in K_{{\mathcal T}}$. 
	Thus, we assume that there exists $m\in{\mathbb N}$ such that for a.~e.\
	$z \in S_0$ with $f(z) \neq 0$ we have $L(z) \geq \frac 1m$ and
	$E(z) \leq E_0 - \frac 1m < 0$.
	
	\noindent
	\textit{Step 2: Approximation like in Proposition~\ref{transapprox}.} 
	We can construct an approximating
	sequence $(F_k)_{k\in{\mathbb N}} \subset C^{\infty}_{c,r} (S_0)$ such that
	\[
	F_k \to f \ \mbox{and}\ {\mathcal T} F_k \to {\mathcal T} f = 0 \ \mbox{in}\ H
	\ \mbox{as}\ k\to\infty,
	\]
	where
	\[
	{\mathrm{supp}\,} F_k \subset \left\{z \in S_0 \mid L(z) \geq \frac 1{2m} \right\} ,\quad
	k\in{\mathbb N}.			
	\]
	Furthermore, $F_k^r \in C^{\infty}_c (S_0^r)$ for every $k\in{\mathbb N}$.
	
	\noindent
	\textit{Step 3: An auxiliary estimate.} 
	In order to prove that the distance between $F_k$ and its projection
	${\mathcal P} F_k$ tends to zero as $k\to\infty$ we first estimate
	the distance between a smooth function and its projection onto the
	space of functions depending only on $E$ and $L$.
	Let $\xi \in C^1_{c,r} (S_0)$ with $\xi^r \in C^1_c (S_0^r)$ be arbitrary,
	but fixed. We will use the abbreviation 
	\[
	\zeta (t,E,L) := \xi^r ( (R,W)(t \, T(E,L), r_-(E,L), 0, L),L )
	\]
	for $t\in{\mathbb R}$, $L>0$ and $\psi_L (r_L) < E < E_0$.
	For these $(E,L)$ we may therefore write
	\[
	{\mathcal P} \xi (E,L) = \int_0^1 \zeta (t,E,L) \,dt.
	\]
	By a slight abuse of notation, we denote the mapping
	\[
	(x,v) \mapsto {\mathcal P} \xi (E(x,v), L(x,v)),	
	\]
	which is defined a.~e.\ on $S_0$, by ${\mathcal P} \xi$.
	By changing to $(t,E,L)$-coordinates,
	\[
	\begin{split}
	&
	\| \xi - {\mathcal P} \xi \|_H^2 \\
	&
	= 4 \pi^2 \int_0^{\infty} \int_{\psi_L(r_L)}^{E_0}
	\frac{T(E,L)}{\vert \phi' (E) \vert}
	\int_0^1 \left| \zeta(t,E,L) - \int_0^1 \zeta (s,E,L) \,ds \right|^2 \,dt
	\,dE \,dL\\
	&
	\leq  4 \pi^2 \int_0^{\infty} \int_{\psi_L(r_L)}^{E_0}
	\frac{T(E,L)}{\vert \phi' (E) \vert}
	\int_0^1 \int_0^1 \left| \zeta(t,E,L) -  \zeta (s,E,L) \right|^2 \,ds \,dt \,
	dE \,dL,
	\end{split}
	\]
	where we used the Cauchy-Schwarz inequality in the last step.
	To estimate the inner two integrals, we first consider the case
	$s \leq t$:
	\[
	\begin{split}
	\int_0^1 \int_0^t & \left| \zeta(t,E,L) -  \zeta (s,E,L) \right|^2 \,ds \,dt
	= \int_0^1 \int_0^t \left| \int_s^t \partial_{\tau} \zeta (\tau,E,L) \,d\tau
	\right|^2 \,ds \,dt \\
	& \leq \int_0^1 \int_0^t (t-s)
	\int_s^t \left| \partial_{\tau} \zeta (\tau,E,L) \right|^2 \,d\tau  \,ds \,dt
	\leq \int_0^1 \left| \partial_{\tau} \zeta (\tau,E,L) \right|^2 \,d\tau,
	\end{split}
	\]
	where we again used the Cauchy-Schwarz inequality.
	Estimating the part where $s > t$ in a similar way, we obtain
	\[
	\| \xi - {\mathcal P} \xi \|_H^2
	\leq 8 \pi^2 \int_0^{\infty} \int_{\psi_L(r_L)}^{E_0}
	\frac{T(E,L)}{\vert \phi' (E) \vert}
	\int_0^1 \left|  \partial_{\tau} \zeta (\tau,E,L) \right|^2 \,d\tau \,dE \,dL .
	\]
	By definition
	\[
	\partial_{\tau} \zeta (\tau,E,L) = T(E,L) \,
	( {\mathcal T}^r \xi^r ) ( (R,W)(\tau \, T(E,L), r_-(E,L), 0, L),L )
	\]
	for $\tau \in{\mathbb R}$, $L>0$, and $\psi_L (r_L) < E < E_0$.
	Using \eqref{Test} we arrive at the estimate
	\[
	\| \xi - {\mathcal P} \xi \|_H^2    
	\leq  32 \pi^4  M_0^4 \int_0^{\infty} \int_{\psi_L(r_L)}^{E_0}
	\frac{T(E,L)}{\vert \phi' (E) \vert}
	\frac 1 {E^4 L} \int_0^1 \left|  ({\mathcal T} \xi)^r (\ldots) \right|^2 d\tau
	\,dE \,dL,
	\]
	where $\ldots$ stands for
	$(R,W)(\tau \, T(E,L), r_-(E,L), 0, L),L$.
	
	\noindent
	\textit{Step 4: Conclusion.} 
	We apply the estimate from the previous step to the elements of the
	approximating sequence. Due to the properties of their support
	we obtain the bound
	\[
	\frac 1{E^4 L} \leq \frac {2m}{E_0^4}
	\]
	for all $L> 0$ and $\psi_L (r_L) \leq E < E_0$ for which there exists
	$\tau \in{\mathbb R}$ and $k\in{\mathbb N}$ such that
	$0 \neq ({\mathcal T} F_k)^r ( (R,W)(\tau \,T(E,L), r_-(E,L), 0, L),L )$.
	Using this bound and changing back into $(x,v)$-coordinates
	we arrive at the estimate
	\[
	\begin{split}
	\| F_k - {\mathcal P} F_k \|_H^2
	&\leq
	64 \pi^4  M_0^4 \frac m{E_0^4}
	\int_0^{\infty} \int_{\psi_L(r_L)}^{E_0}\!
	\frac{T(E,L)}{\vert \phi' (E) \vert}
	\int_0^1 \!\left|  ({\mathcal T} F_k)^r (\ldots) \right|^2 d\tau \,dE \,dL \\
	&=
	16 \pi^2 M_0^4 \frac{m}{E_0^4} \| {\mathcal T} F_k \|_H^2 \to 0\ \mbox{as}\ k\to\infty .
	\end{split}
	\]
	Since $F_k \to f$, we obtain ${\mathcal P} F_k \to f$ in $H$
	for $k\to\infty$ as well. Since $K_{{\mathcal T}}$ is a closed subspace of $H$
	and $({\mathcal P} F_k)_{k\in{\mathbb N}} \subset K_{{\mathcal T}}$,
	we conclude that $f\in K_{{\mathcal T}}$,
	and the proof is complete. 
\end{proof}

\end{document}